\documentclass[12pt]{article}
\usepackage{amsmath}
\usepackage{amsfonts}
\usepackage{amssymb}
\usepackage{amsthm}
\usepackage{graphicx}
\usepackage{fullpage}
\usepackage{setspace}
\usepackage{verbatim}
\usepackage{natbib}
\usepackage{appendix}
\usepackage{rotating}
\usepackage{multirow}

\newtheorem{theorem}{Theorem}[section]

\newtheorem{corollary}{Corollary}[section]
\newtheorem{assumption}{Assumption}[section]

\theoremstyle{definition}

\newtheorem{remark}{Remark}[section]

\title{Asymptotic Efficiency Bounds for a Class of Experimental Designs}
\author{Timothy B. Armstrong\thanks{email: timothy.armstrong@usc.edu.  Support
    from National Science Foundation Grant SES-2049765 is gratefully acknowledged.}  \\
University of Southern California}
\date{\today}

\begin{document}

\maketitle

\begin{abstract}
  We consider an experimental design setting in which units are assigned to
  treatment after being sampled sequentially from an infinite population.
  We derive asymptotic efficiency bounds that apply to data from any experiment
  that assigns treatment as a (possibly randomized) function of covariates and
  past outcome data,
  including stratification on covariates and adaptive designs.
  For estimating the average treatment effect of a binary treatment,
  our results show that
  no further first order asymptotic efficiency improvement is possible relative
  to an estimator that achieves the \citet{hahn_role_1998} bound in an
  experimental design where the propensity
  score is chosen to minimize this bound.
  Our results also apply to settings with multiple treatments with possible
  constraints on treatment, as well as covariate based sampling of a single
  outcome.

\end{abstract}

\section{Introduction}

It is common practice in the design of experiments to use baseline covariates or
data from past waves to inform sampling or treatment assignment.
An example is stratification, in which units are grouped into blocks using
baseline covariates, and then randomized to treatment or control separately
within each block, thereby ensuring that the covariate distribution is
``balanced'' between treatment and controls.
In a review of a selection of research articles using experiments in
development economics, \citet{bruhn_pursuit_2009} report that about $3/4$ of
these articles use some form of stratification.
Further description and discussion of such designs are given in
survey articles \citep[][]{duflo_chapter_2007}
and textbooks
\citep[][]{imbens_causal_2015,rosenberger_randomization_2015}.
See also
\citet{bugni_inference_2018} for further references.

Such designs have received renewed interest in the theoretical literature, with
several papers deriving asymptotic approximations to the sampling distribution
of estimators and test statistics in such designs \citep[see, among
others,][]{bugni_inference_2018,bai_inference_2021}.
One goal of this literature has been to design experiments that improve the
asymptotic efficiency of estimators and tests.
Consider estimating the average treatment effect (ATE) of a binary treatment.
For a given experimental design leading to iid data,
the efficiency bound of \citet{hahn_role_1998} gives the smallest possible
asymptotic variance.
This optimal asymptotic variance depends on the
experimental protocol through the propensity score (the conditional probability
of treatment given covariates).  A recent literature
\citep{hahn_adaptive_2011,tabord-meehan_stratification_2023,cytrynbaum_optimal_2023}
has considered the problem of implementing an
experiment that makes the variance bound of \citet{hahn_role_1998} as small as
possible.

A key component of the experimental designs proposed in this recent literature
is the use of data from past waves or stratification on baseline covariates as a
part of the rule for assigning treatment.
Due to the non-iid nature of the resulting data, the \citet{hahn_role_1998}
bound no longer applies.
Thus, it is unclear whether the \citet{hahn_role_1998}
bound is still an appropriate benchmark once such experimental protocols are
allowed.
Does the \citet{hahn_role_1998} bound still give the optimal asymptotic variance
in such settings?  Or does the use of stratification or other
dependence-inducing experimental designs allow for further improvement that the
aformentioned literature is leaving on the table?

In this paper, we derive asymptotic efficiency bounds in a general setting that
allows for such designs.  Applied to the case of a binary treatment, our results
show that the optimized \citet{hahn_role_1998} bound indeed gives a bound
for the performance of any estimator or test with data from any experimental
design in this general setting.
Thus, no further efficiency improvement is possible using other experimental
designs.

To see how our main result is derived, recall that semiparametric efficiency
bounds such as those derived by \citet{hahn_role_1998} are obtained by deriving
bounds in a parametric submodel.  The sharpest
bound that can be obtained from a parametric submodel
gives the
semiparametric efficiency bound, and the submodel that yields this bound is
called the least favorable submodel.
Thus, to show that the bound continues to hold under arbitrary experimental
designs, it suffices to show that the same efficiency bound holds in the least
favorable submodel.
To this end, we derive
a likelihood expansion and local asymptotic
normality theorem that applies to arbitrary experimental designs that assign
treatment after observing the entire set of covariates and past outcome values
for an independent sample from an infinite population.
To derive these results, we apply techniques used in the recent literature
deriving asymptotic distributions of estimators in related settings \citep[in
particular, we apply a martingale representation similar to those used in
][]{abadie_martingale_2012} to a Le Cam style local expansion of
the likelihood ratio.
Applying these results to the least favorable submodels used to derive the
corresponding bounds in the iid case then gives the efficiency bounds.

As we discuss in more detail in Section \ref{achieving_bound_sec}, one can
achieve the
\citet{hahn_role_1998} variance bound either through independent treatment assignment or through
treatment assignment rules that use
stratification on baseline covariates.
Given that one can achieve the same asymptotic variance using these either approach,
one may use other criteria to choose between different
treatment rules that yield the same optimal asymptotic variance.
For example, one may opt for stratification over iid sampling because it allows
one to achieve the optimal asymptotic variance using a simpler estimator.
We discuss practical implications and limitations of our results and other
results from the literature in Section \ref{recommendations_for_practice_sec}.

Several papers written around the same time as this one consider related
problems involving asymptotic efficiency bounds in experiments.
\citet{bai_efficiency_2023} and \citet{rafi_efficient_2023} consider a setting
similar to ours, but consider efficiency among certain restricted classes of
treatment rules involving covariate based stratification.  This differs from our main efficiency bounds
(Theorems
\ref{ATE_LAN_thm} and \ref{general_LAN_thm})
which do not
restrict the treatment rule or impose only cost constraints,
although our likelihood expansion and general local asymptotic normality result
(Theorem \ref{lr_expansion_thm} and Corollary \ref{lr_clt_corollary}) are useful
as technical tools in these other settings.
Another literature
\citep{adusumilli_risk_2023,kuang_weak_2023,hirano_asymptotic_2023} focuses on bandit problems
and related settings.  While these papers consider interesting dynamic problems
that fall outside of the scope of the present paper (for example, deciding when
to end an experiment early in the interest of the welfare of experimental
subjects), they do not address whether
experimental design choices such as stratified randomization can be used to
improve on efficiency bounds for iid data.

The rest of this paper is organized as follows.
Section \ref{informal_description_sec} gives an informal description of our results in a simple setting
with a binary treatment and no constraints on the experimental design.
Section \ref{main_results_sec} describes the formal setup, and includes our main technical
results.
Section \ref{binary_ate_formal_sec} applies these results to provide a formal statement of the
optimality result in the simple setting in Section \ref{informal_description_sec}.
Section \ref{general_te_section} considers a more general setting with multiple
treatments and possible constraints on overall treatment and sampling.
Section \ref{achieving_bound_sec} discusses approaches for achieving the
efficiency bounds derived in earlier sections.
Section \ref{recommendations_for_practice_sec} discusses practical implications
and limitations of the results.
Proofs are given in an appendix.

\section{Informal Description of Results in a Simple Case}\label{informal_description_sec}

Consider the case of a binary treatment.
Unit $i$ has potential outcomes $Y_i(1)$ and $Y_i(0)$ under treatment and
non-treatment.
In addition, there is a vector of baseline covariates $X_i$ associated with
individual $i$.
We assume that $(X_i,Y_i(0),Y_i(1))$ are
drawn iid from some population, and we are interested in the ATE
$E[Y_i(1)-Y_i(0)]$ for this population.
The researcher first observes a sample $X_1,\ldots,X_n$ of baseline covariates.  The researcher chooses
a treatment assignment $W_{n,i}$ for each unit $i$, and observes $Y_i(W_{n,i})$ for this
unit.  The treatment assignment $W_{n,i}$ can depend on the entire sample of
baseline covariates, as well as past outcomes $Y_j(W_{n,j})$ for
$j=1,\ldots,i-1$.\footnote{We subscript by $n$ as well as $i$ since the
treatment assignment rule depends on the entire sample $X_1,\ldots,X_n$ and can
therefore vary arbitrarily with $n$; see Section \ref{main_results_sec} for a formal
description of our notation.}

One possible design is to assign treatment independently across $i$, with
$P(W_{i}=1|X_i)=e(X_i)$.  The conditional treatment probability $e(x)$ is
referred to in the literature as the propensity score.  This yields iid data, so
that the semiparametric efficiency bound of \citet{hahn_role_1998} applies,
giving
\begin{align}\label{hahn_seb_var_eq}
  v_{e(\cdot)} = var\left( \mu(X_i,1)-\mu(X_i,0) \right)
  +E\frac{\sigma^2(X_i,0)}{1-e(X_i)}
  +E\frac{\sigma^2(X_i,1)}{e(X_i)}
\end{align}
as a bound for the asymptotic variance of an estimator of the ATE, where
$\mu(x,w)=E[Y_i(w)|X_i=x]$ and $\sigma^2(x,w)=var(Y_i(w)|X_i=x)$.
We can choose the propensity score $e(\cdot)$ to minimize this bound by taking
first order conditions: the optimal propensity score $e^*(\cdot)$ satisfies
\begin{align}\label{neyman_allocation_eq}
  \frac{\sigma^2(x,0)}{[1-e^*(x)]^2}
  = \frac{\sigma^2(x,1)}{e^*(x)^2}.
\end{align}
Following the literature, we refer to this as the Neyman allocation, after
\citet{neyman_two_1934}.

Since $e^*()$ requires knowledge of the unknown conditional
variance $\sigma^2(x,w)$, this design is not feasible.  However, a recent literature 
has explored feasible approaches to designing an experiment and estimator that
minimizes the asymptotic variance in (\ref{hahn_seb_var_eq}).  These
experimental designs incorporate the following steps:
\begin{itemize}
\item[Step 1:] Designate the first part of the sample (say, the first
$n_{\operatorname{pilot}}$ observations) as a pilot sample and use these
observations to obtain a preliminary estimate
$\hat\sigma^2_{\operatorname{pilot}}(x,d)$ of the conditional variance function
$\sigma^2(x,d)$.
Formally, the pilot
 size $n_{\operatorname{pilot}}$ must be chosen so that $n_{\operatorname{pilot}}/n\to
 0$ and $n_{\operatorname{pilot}}\to\infty$.

\item[Step 2:] Plug this estimate
$\hat\sigma^2_{\operatorname{pilot}}(x,d)$ into the formula (\ref{neyman_allocation_eq}) to
get an estimate $\hat e^*_{\operatorname{pilot}}(x)$ of the Neyman allocation.

\item[Step 3:] Use the estimate $\hat e^*_{\operatorname{pilot}}(x)$ to assign
treatment to the remaining observations $n_{\operatorname{pilot}}+1,\ldots,n$:
after observing $X_i$, the probability of assigning unit $i$ to treatment is
$\hat e^*_{\operatorname{pilot}}(X_i)$.
\end{itemize}

Using the resulting data, one can then achieve the optimized bound
$v_{e^*(\cdot)}$ using an ATE estimator that flexibly adjusts for covariates or
uses a flexible estimate of the propensity score
\citep{hahn_adaptive_2011}.\footnote{The formal results in
  \citet{hahn_adaptive_2011} focus on the case where
  $n_{\operatorname{pilot}}/n$ converges to a positive constant $\kappa$, with
  $v_{e^*(\cdot)}$ obtained in the limiting case as $\kappa\to 0$.}  For example, if $X_i$ takes
on a finite number of values, one can form the estimate $\hat
\mu(x,d)=\frac{\sum_{i:X_i=x,D_i=d} Y_i}{\# \{i:X_i=x,D_i=d\}}$ and form the estimator
$\frac{1}{n}\sum_{i=1}^n [\hat \mu(X_i,1)-\hat \mu(X_i,0)]$.
Alternatively, if one incorporates stratification on $X_i$ in the treatment
assignment in Step 3, then one can achieve the optimized bound $v_{e^*(\cdot)}$
without the need to flexibly adjust for covariates in the estimation step
\citep{tabord-meehan_stratification_2023,cytrynbaum_optimal_2023}.
We provide a more detailed discussion of these two strategies in Section \ref{achieving_bound_sec}.

Such designs, however, lead to dependent data that violates the assumptions used
in the \citet{hahn_role_1998} bound,
making it unclear whether this bound remains an appropriate benchmark once
non-iid treatment assignments are allowed.
One of the main contributions of this paper is to show that the
variance bound $v_{e^*(\cdot)}$ still applies to these designs, as well as any other experimental
design for assigning treatment as a function of past values and the entire
vector of baseline covariates.  Thus, the combinations of
estimators and experimental designs in
\citet{hahn_adaptive_2011,tabord-meehan_stratification_2023,cytrynbaum_optimal_2023}
are indeed asymptotically optimal among any such design with any possible
estimator.

Formally, semiparametric efficiency bounds amount to a statement that no uniform
efficiency improvement is possible over a class of distributions that is rich
enough to include a particular one dimensional submodel, called a ``least
favorable submodel.''  Our results show that this statement continues to hold
for any experimental design in our setup, with the same least favorable submodel
as in the iid case.

\section{Setup and Main Results}\label{main_results_sec}

This section presents our formal setup and main technical results.
Section \ref{setup_subsection} presents notation and sampling assumptions.
Section \ref{submodel_subsection} presents the assumptions on parametric submodels.
Section \ref{lan_subsection} presents our main likelihood expansion and local asymptotic
normality theorem.

\subsection{Setup and Sampling Assumptions}\label{setup_subsection}

We consider a setting in which baseline covariates $X_i$ and potential outcomes
$\{Y_i(w)\}_{w\in\mathcal{W}}$ are associated with unit $i$, where
$\mathcal{W}$ is a finite set of
possible treatment assignments.
We assume that $X_i,\{Y_i(w)\}_{w\in\mathcal{W}}$ are drawn iid from some
population.
The researcher chooses a treatment assignment $W_{n,i}$ for each observation
$i$, and observes $X_i$ and $Y_{n,i}=Y_i(W_{n,i})$ for each observation $i$.
In forming this assignment rule, the researcher first observes the entire sample
$X^{(n)}=(X_1,\ldots,X_n)$ of covariates.  The rule is then allowed to depend
sequentially on observed outcome variables.  Let
$Y_n^{(i)}=(Y_{n,1},\ldots,Y_{n,i})$.  The treatment rule is given by
$W^{(n)}=(W_{n,1},\ldots,W_{n,n})$ where $W_{n,i}=w_{n,i}(X^{(n)},Y_n^{(i-1)},U)$ is a measurable function of $(X^{(n)},Y_n^{(i-1)},U)$
and $U$ is a random variable independent of the sample, which allows for
randomized treatment rules.
Based on this data, the researcher then forms an estimator or test for some
parameter of the population distribution of $X_i,\{Y_i(w)\}_{w\in\mathcal{W}}$.

We use the convention of labeling treatment groups $w$ by nonnegative integers,
although we do not require that the treatment groups have any natural ordering.
We will also consider settings where one allows
for unit $i$ not to be assigned to any treatment group, in
which case we include the treatment group $w=-1$ in the set $\mathcal{W}$ of
possible treatments and we use the convention that $Y_i=Y_i(-1)=0$ when we set $W_{n,i}=-1$.

\begin{remark}
Our setup allows for experimental designs that use information on baseline
covariates in essentially arbitrary ways.
Designs
involving stratified randomization on covariates and, in particular, matched
pairs, are allowed.
Our setup also includes designs that use outcomes from a pilot study, by
defining observations $1,\ldots,n_{\text{pilot}}$ as observations from this
study.
Note that treating the randomization device $U$ as a random variable of fixed
dimension does not lead to a loss of generality, since transformations of $U$
can be incorporated into the sampling rule $w_{n,i}(X^{(n)},Y_n^{(i-1)},U)$.
\end{remark}

\begin{remark}\label{finite_population_remark}
  We follow much of the literature by assuming that our sample is taken
  independently from an infinite population.
  In particular, this assumption is made in papers deriving asymptotics for
  estimators and tests under stratified sampling including
  \citet{bugni_inference_2018} and \citet{bai_inference_2021}, and papers on
  experimental design including \citet{imbens_finite_2009},
  \citet{hahn_adaptive_2011}, \citet{tabord-meehan_stratification_2023} and
  \citet{cytrynbaum_optimal_2023}.
  One can consider this an
  approximation to a setting where one samples from a large population of $N$
  units.  Formally, each unit $j=1,\ldots, N$ has covariates and outcomes
  $X_j^*,\{Y_j^*(w)\}_{w\in\mathcal{W}}$, and we 
  we draw $X_i,\{Y_i(w)\}_{w\in\mathcal{W}}$ by drawing a random variable $j(i)$
  over the uniform distribution on $1,\ldots, N$, and then defining
  $X_i=X_{j(i)}^*$ and $Y_i(w)=Y_{j(i)}^*(w)$ for each $w\in\mathcal{W}$.  This
  corresponds exactly to sampling from the larger population with replacement, which is
  a good approximation to sampling without replacement when $N$ is large
  relative to $n$.

  Thus, our setup incorporates an assumption that the experimental
  design involves randomized sampling from a large population.\footnote{This also means
  that treatment assignments that assign units to treatment groups
  deterministically as a function of the index $i$ or covariates $X_i$ are still
  ``randomized'' in the sense that the subset of units in each treatment group
  is random as a subset of the larger population.
  For example, the assignment that takes $W_{n,i}=0$ for
  $i=1,\ldots,n/2$ and $W_{n,i}=1$ for $i=n/2+1,\ldots,n$ is
  ``randomized'' in the sense that the sample of treated units $\{j(i):
  i=n/2+1,\ldots,n\}$ is a random subset of the population $1,\ldots,N$, as well
  as being a random subset of the sampled units (it is not a deterministic
  function of the set $\{j(i): i=1,\ldots,n\}$ of sampled units).}
  Results that explicitly address the
  question of whether it is indeed optimal to randomly sample from a (possibly
  large) finite  population include \citet[][Ch. 14, Section
  8]{savage_foundations_1972} and \citet[][Section
  8.7]{blackwell_theory_1954}.\footnote{The notion of  ``optimality'' is slightly different in these references, since they
    consider finite-sample minimax over a fixed set of distributions, in contrast to
    the semiparametric results in the present paper which correspond to
    asymptotic minimax bounds over a localized parameter space.}
  We note that our results do allow for some statements about the optimal use of
  covariates for sampling a single outcome (by taking the set of treatments to be a
  singleton and incorporating cost constraints; see Section
  \ref{survey_sampling_example_sec}).

\end{remark}

\subsection{Parametric Submodel and Likelihood Ratio}\label{submodel_subsection}

We consider a finite dimensional parametric model indexed by $\theta$.  We are
interested in efficiency bounds at a particular $\theta^*$.
While our analysis
will allow us to consider parametric settings, we will be primarily interested
in using least favorable submodels to derive semiparametric efficiency bounds in
infinite dimensional settings, as in the ATE bound for binary treatment
described in Section \ref{informal_description_sec}.
In cases where ambiguity may arise, we subscript expectations $E_\theta$ and
probability statements $P_\theta$ by $\theta$ to indicate that
$X_i,\{Y_i(w)\}_{w\in\mathcal{W}}$ are drawn from this model.

Let 
$f_X(x;\theta)$ denote the density of $X_i$ with respect to
$\nu_X$, and let
$f_{Y(w)|X}(y|x;\theta)$ denote the density of $Y_i(w)$ with respect to $\nu_{Y,w}$,
where $\nu_X$ and $\nu_{Y,w}$ are measures that
do not depend on $\theta$.
Let $p_U$ denote the density of $U$ (which does not depend on $\theta$).
The probability density of $U,X_1,\ldots,X_n,Y_{n,1},\ldots,Y_{n,n}$ is
\begin{align}\label{likelihood_eq}
p_U(u)\prod_{i=1}^n\left[ f_X(x_i;\theta) \prod_{w\in\mathcal{W}} f_{Y(w)|X}(y_i|x_i;\theta)^{I(w_{n,i}=w)}
  \right] 
\end{align}
where $w_{n,i}=w_{n,i}(x_1,\ldots,x_n,y_1,\ldots,y_{i-1},u)$.\footnote{The
  likelihood can be derived recursively by noting that, for any
  $w\in\mathcal{W}$ and any $X^{(n)},Y_n^{(i-1)},U$ such that
  $w_{n,i}(X^{(n)},Y_n^{(i-1)},U)=w$, we have $Y_i=Y_i(w)$ and
  $Y_i(w)|W_{n,i}=w,X^{(n)},Y_n^{(i-1)},U
  \stackrel{d}{=}Y_i(w)|X_i$
  by the independence assumptions.}
The researcher makes a decision using the observed data
$X_1,\ldots,X_n,Y_{1,n},\ldots,Y_{n,n}$, along with the treatment rule and the variable
$U$, which determine the treatment assignments $W_{i,n}$.
Since the treatment rule is known once $U$ is given, we can take the observed
data to be $X_1,\ldots,X_n,Y_{1,n},\ldots,Y_{n,n}$ and $U$, so that the
likelihood is given by (\ref{likelihood_eq}).

Following the literature on asymptotic efficiency, we make a quadratic mean
differentiability assumption on the model \citep[see][Section 7.2, for a definition]{van_der_vaart_asymptotic_1998}.

\begin{assumption}\label{qmd_assump}
The family $f_X(x;\theta)$ is differentiable in quadratic mean (qmd) at
$\theta^*$ with score function $s_X(X_i)$, and, for each $w\in\mathcal{W}$, the family
$f_{Y(w)|X}(y|x;\theta)$ is qmd at $\theta^*$ with score function $s_w(Y_i(w)|X_i)$.
\end{assumption}
Here, the qmd condition for the conditional distribution
$f_{Y(w)|X}(y|x,\theta)$ is taken to mean that the family is qmd when $X_i$ is
distributed according to $\theta^*$; i.e. the family $\theta\mapsto f_X(x;\theta^*)f_{Y(w)|X}(y|x;\theta)$ is qmd at $\theta^*$.
Let $I_X=E_{\theta^*}s_X(X_i)s_X(X_i)'$ denote the information for $X_i$,
and let
$I_{Y(w)|X}(x)=E_{\theta^*}[s_{w}(Y_i(w)|X_i)s_{w}(Y_i(w)|X_i)'|X_i=x]$
and
$I_{Y(w)}=E_{\theta^*}I_{Y(w)|X}(X_i)=E_{\theta^*}[s_{w}(Y_i(w)|X_i)s_{w}(Y_i(w)|X_i)']$
denote the conditional and unconditional information for $Y_i(w)$ for each $w$.
Note that these are finite by Theorem 7.2 in \citet{van_der_vaart_asymptotic_1998}.

\subsection{Likelihood Expansion and Local Asymptotic Normality}\label{lan_subsection}

Consider a sequence $\theta_n=\theta^*+h/\sqrt{n}$ where 
$\theta^*$ is given.
To obtain efficiency bounds, we extend Le Cam's result on the
asymptotics of likelihood ratio statistics in parametric families
(Theorem 7.2 in \citet{van_der_vaart_asymptotic_1998})
to our setting, with the likelihood given in (\ref{likelihood_eq}).
Since $p_{U}$
does not depend on $\theta$, this term drops out, and the log of the likelihood
ratio for $\theta^*$ vs $\theta_n$ is given by
\begin{align*}
  \ell_{n,h}=\sum_{i=1}^n \tilde\ell_X(X_i;\theta_{n})
  + \sum_{w\in\mathcal{W}}\sum_{i=1}^n I(W_{n,i}=w)\tilde\ell_{Y(w)|X}(Y_i,X_i;\theta_{n})
\end{align*}
where
\begin{align*}
  &\tilde \ell_X(x;\theta)\equiv \log \frac{f_{X}(x;\theta)}{f_{X}(x;\theta^*)},
  \quad
  \tilde \ell_{Y(w)|X}(y,x;\theta)\equiv \log \frac{f_{{Y(w)|X}}(y;x,\theta)}{f_{{Y(w)|X}}(y;x,\theta^*)}, \, w\in\mathcal{W}.
\end{align*}

\begin{theorem}\label{lr_expansion_thm}
  Under Assumption \ref{qmd_assump}, the likelihood ratio $\ell_{n,h}$ satisfies
  \begin{align}\label{lr_expansion_eq}
  \ell_{n,h}
  &=\frac{1}{\sqrt{n}}\sum_{i=1}^n h's_X(X_i)
  + \frac{1}{\sqrt{n}}\sum_{i=1}^n\sum_{w\in\mathcal{W}}I(W_{n,i}=w)h's_{w}(Y_i(w)|X_i)  \nonumber  \\
  &- \frac{1}{2}h'I_X h
  - \frac{1}{2n}\sum_{i=1}^n \sum_{w\in\mathcal{W}}I(W_{n,i}=w)h'I_{Y(w)|X}(X_i)h
    + o_{P_{\theta^*}}(1).
  \end{align}
\end{theorem}

Theorem \ref{lr_expansion_thm} can be used to prove the following local asymptotic normality result.

\begin{corollary}\label{lr_clt_corollary}
Suppose Assumption \ref{qmd_assump}, holds and let
$\tilde I_n=I_X+\frac{1}{n}\sum_{i=1}^n \sum_{w\in\mathcal{W}}I(W_{n,i}=w)I_{Y(w)|X}(X_i)$.
Let $\tilde I^*$ be a positive definite symmetric matrix.
\begin{itemize}
\item[i.)] If $\tilde I_n$ converges in probability to $\tilde I^*$ under $\theta^*$, then
$\ell_{n,h}$ converges in distribution to a $N(-h'\tilde I^* h/2,h'\tilde I^*
h)$ law under $\theta^*$.
\item[ii.)] If $\tilde I_n\le \tilde I^*+o_{P_\theta^*}(1)$ (where inequality is in the
positive definite sense), then one can define
a probability space under each
$\theta$ with an additional random variable $Z^{(n)}$ (and with the marginal
distribution of $U,X^{(n)},Y_n^{(n)}$ under $\theta$ unchanged) such that
$\tilde \ell_{n,h}=\log \frac{dP_{\theta^*+h/\sqrt{n}}}{dP_{\theta^*}}(U,X^{(n)},Y_n^{(n)},Z^{(n)})$ converges in
distribution to a $N(-h'\tilde I^* h/2,h'\tilde I^*h)$ law under $\theta^*$.
\end{itemize}
\end{corollary}

According to Corollary \ref{lr_clt_corollary}, the model indexed by $\theta^*+h/\sqrt{n}$ is
locally asymptotically normal in the sense of Definition 7.14 in  \citet{van_der_vaart_asymptotic_1998}.  Therefore, the risk of any decision
is bounded from below asymptotically by the risk from a decision in the limiting
model, in which a $N(h,\tilde I^*)$ random variable is observed.
Note that part (ii) of Corollary \ref{lr_clt_corollary} involves augmenting the
data by additional random variables $Z_i$ to handle the case where $\tilde I_n$
may not converge in probability.
This step is a technical trick that appears to be needed to cover, for
example, treatment rules that do not assign any treatment to some individuals,
which is relevant in the setting in Section \ref{general_te_section} with cost
constraints. 
The bounds obtained from local asymptotic normality still apply to
the original setting in which the variables $Z_i$ are not observed, since the bound
from the $N(h,\tilde I^*)$ model applies to decisions that do not use the
variables $Z_i$.

\begin{remark}
While the focus of this paper is on obtaining bounds for experiments that
optimize both the treatment assignment rule and estimator, we note that
Theorem \ref{lr_expansion_thm} and Corollary \ref{lr_clt_corollary} can also be
applied to the topic of efficient estimation under treatment assignment rules
that are not necessarily optimal.
In complementary work, \citet{bai_efficiency_2023} apply these results to derive
sharp asymptotic variance bounds for estimation under non-iid treatment
assignment rules that may not be fully optimized.
Such results are relevant when practical considerations make it difficult or
infeasible to implement a treatment assignment rule that is fully optimal.

\end{remark}

\section{Efficiency Bounds for Average Treatment Effect}\label{binary_ate_formal_sec}

We now apply these results to derive the asymptotic efficiency bound for estimation and
inference on the average treatment
effect (ATE) $E[Y_i(1)-Y_i(0)]$ in the case of a binary treatment
($\mathcal{W}=\{0,1\}$), as described in Section \ref{informal_description_sec}.
Given a population distribution, the variance bound (\ref{hahn_seb_var_eq})
corresponds to a least favorable one-dimensional submodel indexed by
$\theta\in\mathbb{R}$, with $\theta^*$ corresponding to the given population
distribution.  
Thus, we consider the variance bound $v_{e()}$ in (\ref{hahn_seb_var_eq}) with
$\mu(x,w)=\mu_{\theta^*}(x,w)=E_{\theta^*}[Y_i(w)|X_i=x]$ and
$\sigma^2(x,w)=\sigma_{\theta^*}^2(x,w)=var(Y_i(w)|X_i=x)$,
and we define the Neyman allocation $e^*(x)$ in (\ref{neyman_allocation_eq})
with $\sigma^2(x,w)=\sigma_{\theta^*}^2(x,w)=var(Y_i(w)|X_i=x)$.
We then consider a submodel through $\theta^*$ that corresponds to the least
favorable submodel used to derive this bound in the iid case.  Calculations in
\citet[][pp. 326-327]{hahn_role_1998} show that this submodel takes the form in
Section \ref{main_results_sec}, with
\begin{align}\label{ate_lf_submodel_eq}
&s_X(X_i)=\mu_{\theta^*}(X_i,1)-\mu_{\theta^*}(X_i,0)-E_{\theta^*}[\mu_{\theta^*}(X_i,1)-\mu_{\theta^*}(X_i,0)],  \nonumber  \\
&s_{0}(Y_i|X_i)=\frac{Y_i(0)-\mu_{\theta^*}(X_i,0)}{1-e(x_i)}
\quad\text{and}\quad
s_{1}(Y_i|X_i)=\frac{Y_i(1)-\mu_{\theta^*}(X_i,1)}{e(x_i)}.
\end{align}
The score function for this submodel is
\begin{align*}
s(X_i,Y_i(0),Y_i(1),W_i) = s_X(X_i) + (1-W_i)s_{0}(Y_i|X_i) + W_is_{1}(Y_i|X_i)
\end{align*}
and the information is $E_{\theta^*}s(X_i,Y_i(0),Y_i(1),W_i)^2=v_{e(\cdot)}$.
Furthermore, letting $ATE(\theta)=E_{\theta}[Y_i(1)-Y_i(0)]$ for $\theta$ in
this submodel the calculations in \citet[][pp. 326-327]{hahn_role_1998} show that $ATE(\theta)$ is differentiable at $\theta^*$ in the sense of 
p. 363 of \citet{van_der_vaart_asymptotic_1998}, and that
$s(X_i,Y_i(0),Y_i(1),W_i)$ is the efficient influence function, so that
\begin{align}\label{ATE_derivative_eq}
  ATE(\theta^*+t)-ATE(\theta^*)=t E_{\theta^*}s(X_i,Y_i(0),Y_i(1),W_i)^2 + o(t)
  =t v_{e(\cdot)} + o(t)
\end{align}
as $t\to 0$.
These calculations require regularity conditions on the submodel so that certain
derivatives can be taken under integrals.  Rather than stating these as
primitive conditions, we will assume (\ref{ATE_derivative_eq}) directly.

We now apply Theorem \ref{lr_expansion_thm} to show that no further improvement
is possible relative to
the semiparametric efficiency bound $v_{e^*()}$, with propensity score
given by the Neyman allocation $e^*()$.  We begin with a local asymptotic
normality theorem.

\begin{theorem}\label{ATE_LAN_thm}
In the binary treatment setting with $\mathcal{W}=\{0,1\}$,
consider
a model satisfying Assumption \ref{qmd_assump}, with $s_X$, $s_{0}$ and
$s_{1}$ given by the score (\ref{ate_lf_submodel_eq}) for the least
favorable submodel with $e(\cdot)$ given by the Neyman allocation
(\ref{neyman_allocation_eq}).  Let 
$w_{n,i}(X^{(n)},Y_{n}^{(i-1)},U)$
be any sequence of treatment rules.
Then the sequence of experiments $P_{\theta^*+h/\sqrt{n}}$ is locally
asymptotically normal \citep[as defined in Definition 7.14, p. 104
of][]{van_der_vaart_asymptotic_1998} with information $v_{e^*(\cdot)}$:
$\ell_{n,h}$ converges in distribution to a $N(-h^2v_{e^*(\cdot)}/2,h^2
v_{e^*(\cdot)})$ law under $\theta^*$.
\end{theorem}

A consequence of the local asymptotic normality result in Theorem \ref{ATE_LAN_thm} and the
differentiability of the ATE parameter in this submodel, as defined in
(\ref{ATE_derivative_eq}), is that the efficiency bound
$v_{e^*(\cdot)}$ gives a bound on the asymptotic performance of any procedure
under any sampling scheme.
We now state a local asymptotic minimax result,
which gives such a bound for estimators in this setting.  Other statements from
asymptotic efficiency theory in regular parametric and semiparametric models (as
in, e.g. Chapters 7, 8, 15 and 25 of \citet{van_der_vaart_asymptotic_1998})
follow as well, but we omit them in the interest of space.

\begin{corollary}\label{ate_asymptotic_minimax_corollary}
Suppose in addition that (\ref{ATE_derivative_eq}) holds.
Let
$\widehat{ATE}_n=\widehat{ATE}_n(X^{(n)},Y_{n}^{(n)},W^{(n)})$
be any sequence of estimators
computed under some sequence of
treatment rules $W_{n,i}=w_{n,i}(X^{(n)},Y_{n}^{(i-1)},U)$. 
For any
loss function $L$ that is subconvex \citep[as defined on p. 113
of][]{van_der_vaart_asymptotic_1998}, we have
\begin{align*}
  \sup_{A}\liminf_{n\to\infty}\sup_{h\in A}
  E_{\theta^*+h/\sqrt{n}} L(\sqrt{n}(\widehat{ATE}_n - ATE(\theta^*+h/\sqrt{n})))
  \ge E_{T\sim N(0,v_{e^*(\cdot)})} L(T)
\end{align*}
where the first supremum is over all finite sets in $\mathbb{R}$.
\end{corollary}

\begin{remark}
  Note that Theorem \ref{ATE_LAN_thm} also implies that, in the least favorable submodel, any treatment
  assignment rule leads to the same optimal variance.  To get some intuition for
  this, we can think of our setting as a game against nature in which the
  researcher chooses an assignment rule and a decision procedure, and nature
  chooses a submodel.  In this game, nature chooses a least favorable submodel,
  which makes the researcher indifferent between all treatment assignments, just
  as an opponent's optimal strategy makes a player indifferent between all
  pure strategies that have positive probability of being played in a mixed
  strategy equilibrium.
  To achieve this, the least favorable submodel sets the information
  $I_{Y(w)|X}(X_i)$ to be equal across the treatment groups $w=0,1$ (note that Theorem
  \ref{ATE_LAN_thm} considers the case where $\mathcal{W}=\{0,1\}$, which rules
  out the possibility of assigning $W_{n,i}=-1$ and excluding unit $i$ from
  either treatment group).
  This leads to the information matrix $\tilde I_n$ defined in Corollary
  \ref{lr_clt_corollary} being the same for any treatment assignment rule.

  Of course, this does not mean that arbitrary treatment assignments can
  be used to achieve this bound in a nonparametric setting.  For example, if one assigns all units to
  treatment, then clearly the ATE cannot even be consistently estimated, since we never
  observe untreated units.  Such
  assignments are optimal in the least favorable submodel, but they can perform
  strictly worse outside of this submodel.
  Again, the analogy of a game against nature is helpful: while the researcher
  is indifferent between certain pure strategies in equilibrium, such pure
  strategies do not themselves constitute equilibrium play.
\end{remark}

\begin{remark}

While Theorem \ref{ATE_LAN_thm} and Corollary
\ref{ate_asymptotic_minimax_corollary} concern a submodel that is least
favorable for a particular treatment allocation, the efficiency bound applies to
any treatment allocation scheme and any estimator.
Furthermore, a submodel of this form will be contained in any
class of distributions $\mathcal{P}$ local to $P_{\theta^*}$ that is
constrained only by regularity conditions such as smoothness conditions and
moment bounds.\footnote{For example, one can construct such a submodel from any $P_{\theta^*}$ by
multiplying the likelihoods $f_X(x;\theta^*)$ and
$\{f_{Y(w)|X}(y|x;\theta^*)\}_{w\in\{0,1\}}$ by an appropriate function of the
score as in \citet[][Example 25.16]{van_der_vaart_asymptotic_1998}.}
Thus, Corollary \ref{ate_asymptotic_minimax_corollary} gives a lower
bound for any estimator under any treatment allocation scheme when one considers
worst-case risk over a class of data generating processes for potential outcomes
that is constrained only by regularity conditions such as smoothness conditions
and moment bounds.
A similar comment applies to
Theorem \ref{general_LAN_thm} and Corollary
\ref{general_asymptotic_minimax_corollary} below.
\end{remark}

\section{Multiple Treatments and Constraints}\label{general_te_section}

We now generalize the setup in Section \ref{binary_ate_formal_sec} to derive
efficiency bounds allowing for multiple treatments and constraints on the number
of units sampled or assigned to each treatment.  Such constraints may arise from
a budget constraint on a costly treatment, or on the overall number of units
sampled.
We first describe the general setup and results (Section \ref{multiple_treatments_general_results_sec}) and then
describe some particular applications that arise as special cases (Section \ref{multiple_treatments_applications_sec}).

\subsection{Setup and Results}\label{multiple_treatments_general_results_sec}

Consider a parameter
\begin{align}\label{general_additive_parameter_eq}
  \tau=\sum_{w\in\mathcal{W}} E[a(X_i,Y_i(w),w)]=\sum_{w\in\mathcal{W}} E[\tilde Y_i(w)].
\end{align}
where $\tilde Y_i(w)=a(X_i,Y_i(w),w)$ for a function $a(x,y,w)$ specified by the
researcher.\footnote{The results in this section can be applied more generally
  to parameters that have the same efficient influence function as a parameter
  that takes the form in (\ref{general_additive_parameter_eq}).  See Remark
  \ref{general_parameter_remark} and the example in Section
  \ref{other_parameters_sec}.}
Consider first a treatment assignment rule in which treatment $w$ is assigned with
probability $p(X_i,w)$ given $X_i$, independently over $i$.  We allow for the
possibility that the treatment probabilities do not add up to one, in which case
we set $W_{n,i}=-1$ and $Y_i=0$ with probability
$1-\sum_{w\in\mathcal{W}}p(X_i,w)$ conditional on $X_i$.
We will show that no further efficiency gain is possible relative to an
estimator that achieves the semiparametric efficiency bound under this
independent sampling scheme with $p()$ chosen to minimize this bound.

The semiparametric efficiency bound for $\tau$ under this sampling
scheme\footnote{The semiparametric efficiency calculations here correspond to a
  known conditional treatment probability $p(X_i,w)$, which reflects the fact
  that $p(X_i,w)$ is known to the experimenter.  However, the efficiency bound
  $v_{p(\cdot)}$ turns out to be the same for this parameter as in the case
  where $p(X_i,w)$ is unknown \citep[see, for example, the discussion in][p. 142]{cattaneo_efficient_2010}.}
at a
distribution corresponding to $\theta^*$ is given by
\begin{align*}
  v_{p(\cdot)}=var_{\theta^*}\left[ \sum_{w\in\mathcal{W}} \tilde\mu_{\theta^*}(X_i,w) \right]
     + \sum_{w\in\mathcal{W}} E_{\theta^*}\frac{\tilde\sigma^2_{\theta^*}(X_i,w)}{p(X_i,w)}
\end{align*}
where $\tilde\mu_{\theta^*}(X_i,w)=E_{\theta^*}[\tilde Y_i(w)|X_i]$ and
$\tilde\sigma^2_{\theta^*}(X_i,w)=var_{\theta^*}(\tilde Y_i(w)|X_i)$. 
The least favorable submodel takes the form in Section \ref{main_results_sec} with
\begin{align}\label{general_lf_submodel_eq}
&s_X(X_i)=\sum_{w\in\mathcal{W}}[\tilde\mu_{\theta^*}(X_i,w)-E_{\theta^*}\tilde\mu_{\theta^*}(X_i,w)]  \nonumber  \\
&s_{w}(Y_i(w)|X_i)=\frac{\tilde Y_i(w)-\mu_{\theta^*}(X_i,w)}{p(X_i,w)},  \quad 
w\in \mathcal{W}
\end{align}
The score function for this submodel is
\begin{align*}
s(X_i,\{Y_i(w)\}_{w\in\mathcal{W}},W_i) = s_X(X_i) + \sum_{w\in\mathcal{W}}I(W_{n,i}=w)s_w(Y_i(w)|X_i).
\end{align*}
Furthermore, letting $\tau(\theta)=\sum_{w\in\mathcal{W}} E_\theta[a(X_i,Y_i(w),w)]=\sum_{w\in\mathcal{W}} E_\theta[\tilde Y_i(w)]$
for $\theta$ in this submodel, 
$\tau(\theta)$ is differentiable at $\theta^*$
in the sense of 
p. 363 of \citet{van_der_vaart_asymptotic_1998}, and 
$s(X_i,\{Y_i(w)\}_{w\in\mathcal{W}},W_i)$ is the efficient influence function, so that
\begin{align}\label{tau_deriv_eq}
  \tau(\theta^*+t)-\tau(\theta^*)=t E_{\theta^*}s(X_i,\{Y_i(w)\}_{w\in\mathcal{W}},W_i)^2 + o(t)
  =t v_{p(\cdot)} + o(t)
\end{align}
as $t\to 0$.
This follows by arguments similar to those in \citet{hahn_role_1998}.
These arguments require regularity conditions on the submodel to ensure that
certain derivatives can be taken under integrals.  Rather than stating these as
primitive conditions, we will assume (\ref{tau_deriv_eq}) directly.

Consider minimizing $v_{p(\cdot)}$ over $p(\cdot)$ subject to constraints
\begin{align}\label{indep_sampling_constraints_eq}
  \sum_{w\in \mathcal{W}} p(x,w) \le 1 \text{ all }x,
  \quad
  \sum_{w\in \mathcal{W}} E_{\theta^*} r(X_i,w)p(X_i,w)\le c
\end{align}
where $c$ is a $d_r\times 1$ vector and $r(\cdot)$ is a $d_r\times 1$ vector
valued function.  The first constraint simply states that treatment
probabilities do not add up to more than one.  The second constrains some linear
combination of overall treatment probabilites.
For example, if $\mathcal{W}=\{0,1\}$ with $1$ corresponding to a costly
treatment, we could take $r(x,w)=I(w=1)$ to incorporate a constraint on overall
cost of the experiment, as in \citet{hahn_adaptive_2011} (see Section \ref{costly_treatment_section} below).
Letting $\lambda(x)$ and $\mu$ be Lagrange multipliers for these constraints and
dropping the first term of $v_{p(\cdot)}$, which does not depend on $p(\cdot)$,
the Lagrangian is
\begin{align*}
  \mathcal{L}
  =E_{\theta^*}\left\{ \sum_{w\in\mathcal{W}} \frac{\tilde\sigma^2_{\theta^*}(X_i,w)}{p(X_i,w)}
  + \lambda(X_i)\left[ \sum_{w\in\mathcal{W}} p(X_i,w) - 1 \right]
  + \mu'\left[ \sum_{w\in\mathcal{W}}r(X_i,w)p(X_i,w) - c \right]\right\}.
\end{align*}
Let $p^*(x,w)$ be the choice of $p(\cdot)$ that solves this problem.
Taking first order conditions gives
\begin{align}\label{p_foc_eq}
  \frac{\tilde\sigma^2_{\theta^*}(x,w)}{p^*(x,w)^2}
  = \lambda(x) + \mu'r(x,w)
   \quad \text{all }x,w.
\end{align}
The complementary slackness conditions are
\begin{align}\label{p_slackness_eq}
  \lambda(x)\sum_{w\in\mathcal{W}}p^*(x,w) = \lambda(x)
  \text{ all } x,
  \quad
  \mu_k \sum_{w\in\mathcal{W}} E_{\theta^*}p^*(X_i,w)r_k(X_i,w)
  = \mu_k c_k
  \,\, k=1,\ldots,d_r.
\end{align}
Note, in particular that, in the least favorable submodel,
$I_{Y(w)|X}(x)=\frac{\tilde\sigma^2_{\theta^*}(x,w)}{p^*(x,w)^2}=\lambda(x) +
\mu'r(x,w)$, and the semiparametric efficiency bound can be written as
\begin{align*}
  &v_{p^*(\cdot)}= I_X + \sum_{w\in\mathcal{W}} E_{\theta^*}p^*(X_i,w)I_{Y(w)|X}(X_i)  \\
  &= I_X + \sum_{w\in\mathcal{W}} E_{\theta^*}p^*(X_i,w)\lambda(X_i)
  + \mu' \sum_{w\in\mathcal{W}} E_{\theta^*}p^*(X_i,w)r(X_i,w)  \\
  &= I_X + E_{\theta^*}\lambda(X_i)
  + \mu' c
\end{align*}
where the last step uses the complementary slackness condition (\ref{p_slackness_eq}).

Now consider the performance of an alternative sampling scheme $w_{n,i}(X^{(n)},Y_{n}^{(i-1)},U)$
under this submodel.
We impose that the constraints (\ref{indep_sampling_constraints_eq}) hold on
average, in the sense that
\begin{align}\label{arbitrary_sampling_constraints_eq}
  \frac{1}{n}\sum_{i=1}^n\sum_{w\in \mathcal{W}} r(X_i,w)I(W_{n,i}=w)\le c + o_{P_{\theta^*}}(1).
\end{align}

\begin{theorem}\label{general_LAN_thm}
  Consider a model satisfying Assumption \ref{qmd_assump} with $s_X$ and $s_w$ given by
  (\ref{general_lf_submodel_eq})
  with $p(\cdot)$ satisfying (\ref{p_foc_eq}) and (\ref{p_slackness_eq}).
  Let $w_{n,i}(X^{(n)},Y_n^{(i-1)},U)$ be any sequence of treatment rules
  satisfying (\ref{arbitrary_sampling_constraints_eq}).
  Then the sequence of experiments $P_{\theta^*+h/\sqrt{n}}$ (possibly modified
  so that it is defined on 
  $X^{(n)},Y^{(n)},U,Z^{(n)}$ where $Z^{(n)}$ is an auxiliary random variable
  and the marginal distribution of $X^{(n)},Y^{(n)},U$ remains unchanged)
  is locally
  asymptotically normal \citep[as defined in Definition 7.14, p. 104
  of][]{van_der_vaart_asymptotic_1998} with information $v_{p^*(\cdot)}$:
  $\log \frac{dP_{\theta^*+h/\sqrt{n}}}{dP_{\theta^*}}(U,X^{(n)},Y_n^{(n)},Z^{(n)})$ converges in distribution to a $N(-h^2v_{p^*(\cdot)}/2,h^2
  v_{p^*(\cdot)})$ law under $\theta^*$.
\end{theorem}

Theorem \ref{general_LAN_thm} and the differentiability condition
(\ref{tau_deriv_eq}) imply that a normal shift experiment with variance
$v_{p^*(\cdot)}$ provides a bound on the performance of any decision and under
any feasible treatment rule in this submodel.
We now provide a formal statement
for estimation in the form of a local asymptotic
minimax theorem.  This generalizes Corollary \ref{ate_asymptotic_minimax_corollary}
to the setting considered in this section.
As with Corollary \ref{ate_asymptotic_minimax_corollary}, we omit other
efficiency statements (such as efficiency bounds for hypothesis tests, or bounds
on the variance of regular estimators) in the interest of space.

\begin{corollary}\label{general_asymptotic_minimax_corollary}
  Suppose, in addition, that (\ref{tau_deriv_eq}) holds.  Let
  $\hat\tau_n=\hat\tau_n(X^{(n)},Y_n^{(n)},W^{(n)})$ be any sequence of estimators
  computed under some sequence of treatment rules
  $W_{n,i}=w_{n,i}(X^{(n)},Y_{n}^{(i-1)},U)$.  For any loss function $L$ that is
  subconvex \citep[as defined on p. 113 of][]{van_der_vaart_asymptotic_1998}, we
  have
  \begin{align*}
    \sup_{A}\liminf_{n\to\infty}\sup_{h\in A}
    E_{\theta^*+h/\sqrt{n}} L(\sqrt{n}(\hat\tau_n - \tau(\theta^*+h/\sqrt{n})))
    \ge E_{T\sim N(0,v_{p^*(\cdot)})} L(T)
  \end{align*}
  where the first supremum is over all finite sets in $\mathbb{R}$.
\end{corollary}

\begin{remark}\label{general_parameter_remark}
  Corollary \ref{general_asymptotic_minimax_corollary} applies to any parameter
  such that, when the parameter $\tau(\theta)$ is defined as a function of
  $\theta$ in the given submodel,
  the differentiability condition (\ref{tau_deriv_eq}) holds.
  In addition to parameters $\tau$ that directly take the form given in
  (\ref{general_additive_parameter_eq}), this includes
  other parameters that share the same form of efficient influence function.
  Quantile treatment effects, which we discuss in Section
  \ref{other_parameters_sec} below, are one example of a parameter that falls
  into this category.
\end{remark}

\subsection{Applications}\label{multiple_treatments_applications_sec}

We now describe several applications that fall into the general setup of Section
\ref{multiple_treatments_general_results_sec}.

\subsubsection{Costly treatment}\label{costly_treatment_section}

Consider the setting in Section \ref{binary_ate_formal_sec} in which we have binary treatment $W_i\in
\{0,1\}$ and we are interested in the average treatment effect
$E[Y_i(1)-Y_i(0)]$.  This falls into the general setup in Section
\ref{multiple_treatments_general_results_sec} with $a(x,y,w)=I(w=1)-I(w=0)$ and
with treatment probability $p(x,1)=e(x)$.  Suppose that, due to budget
constraints, only a fraction $\overline p$ of individuals can be treated.
\citet{hahn_adaptive_2011} allow for such a constraint to be incorporated as a
bound on the overall treatment probability.
In the case of iid treatment
assignment with propensity score $e(\cdot)$, this constraint takes the form in (\ref{indep_sampling_constraints_eq}) with
$r(x,w)=I(w=1)$ and $c=\overline p$:
\begin{align*}
  P(W_i=1)=E e(X_i)\le \overline p.
\end{align*}
Rearranging the first order conditiongs given in (\ref{p_foc_eq}) and plugging
in $p^*(x,1)=e^*(x)$ and $p^*(x,0)=1-e^*(x)$ gives
\begin{align}\label{costly_treatment_focs_eq}
  &\frac{\sigma^2(x,1)}{e^*(x)}=\frac{\sigma^2(x,0)}{(1-e^*(x))}+\mu
\end{align}
where $\mu$ is the Lagrange multiplier on the constraint $Ee(X_i)\le \overline p$.
The optimal iid treatment allocation can then be solved numerically by combining
these first order conditions with the constraint $Ee(X_i)\le \overline p$ to
solve for the Lagrange multiplier $\mu$.
Note that when the constraint doesn't bind, the Lagrange multiplier $\mu$ is
equal to $0$ so that we recover the optimal allocation given in
(\ref{neyman_allocation_eq}) for the unconstrained case.

To make this allocation feasible, \citet{hahn_adaptive_2011} propose to use a
pilot sample to estimate $\sigma^2(x,w)$.  This amounts to the algorithm
described in Section \ref{informal_description_sec}, with the formula (\ref{neyman_allocation_eq})
replaced by the formula for the constrainted optimal treatment assignment given
in (\ref{costly_treatment_focs_eq}).
While \citet{hahn_adaptive_2011} focus on iid treatment allocations, it follows from Theorem
\ref{general_LAN_thm} and Corollary
\ref{general_asymptotic_minimax_corollary} that no further asymptotic
efficiency improvement is possible, so long as the constraint on overall
treatment holds asymptotically in the sense of
(\ref{arbitrary_sampling_constraints_eq}):
\begin{align*}
  \frac{1}{n}\sum_{i=1}^n W_i\le \overline p+o_P(1).
\end{align*}

\subsubsection{Survey sampling}\label{survey_sampling_example_sec}

While our main focus is on experiments, the setup in Section
\ref{multiple_treatments_general_results_sec} can be adapted to the problem of
survey sampling.  Suppose we have an initial iid sample of $n$ units for which we
observe information about some baseline covariates $X_i$.  We are tasked with
choosing a subset of these individuals to conduct a survey asking about some
outcome variable $\tilde Y_i$, with the goal of learning the average value
$E[\tilde Y_i]$ of $\tilde Y_i$ in the population.
In conducting this survey, we have a constraint that
no more than $\overline p\cdot n$ individuals can be surveyed on average, when
individuals are randomly selected for the survey from the initial sample of $n$
individuals.

Let $W_i$ denote an indicator variable equal to $1$ if
individual $i$ is selected for the survey and equal to $-1$ if not.
This fits into our framework with a single potential
outcome $Y_i(1)$ which we denote simply by $\tilde Y_i$ (recall that we allow for
the possibility of setting $W_i=-1$ and not observing any potential outcome, so
that $\mathcal{W}=\{-1,1\}$ in this case).
The constraint on the total survey size can be expressed as the requirement that
$\frac{1}{n}\sum_{i=1}^nI(W_i=1)\le \overline p+o_P(1)$, which is a constraint of the
form (\ref{arbitrary_sampling_constraints_eq}) with $r(X_i,w)=I(W_i=1)$ and
$c=\overline p$.

For an independent sampling scheme in which $W_i=1$ with probability $p(X_i)$
conditional on $X_i$, this constraint takes the form in
(\ref{indep_sampling_constraints_eq}): $Ep(X_i)\le \overline p$.
The optimal asymptotic variance is then given by
$v_{p(\cdot)}=var(\mu(X_i))+E\frac{\sigma^2(X_i)}{p(X_i)}$.
where $\mu(x)=E[\tilde Y_i|X_i=x]$ and $\sigma^2(X_i)=var(\tilde Y_i|X_i=x)$.
Under regularity conditions on $\mu(x)$, one can achieve this variance with the
estimator $\frac{1}{n}\sum_{i=1}^n \hat \mu(X_i)$ where $\hat \mu(x)$ is a
nonparametric estimate of $\mu(x)$ formed from the sample with $W_i=1$.  For
example, in the case where $X_i$ takes on a finite set of values, we can use the
estimate $\hat\mu(x)=\frac{\sum_{i: X_i=x, W_i=1} Y_i}{\#\{i: X_i=x, W_i=1\}}$.
The first order conditions for 
minimizing $v_{p(\cdot)}$ over $p(\cdot)$ subject to the constraints $Ep(X_i)\le
p$ and $p(x)\le 1$ take the form in (\ref{p_foc_eq}):
\begin{align*}
  \frac{\sigma^2(x)}{p(x)^2}=\lambda(x)+\mu.
\end{align*}
The complementary slackness condition states that $\lambda(x)$ is nonzero
only if $p(x)=1$.  It follows that the optimal sampling probability satisfies
\begin{align*}
  p^*(x)=\min\left\{ \sigma(x)/\sqrt{\mu}, 1 \right\}
\end{align*}
where $\mu$ is chosen so that the constraint $Ep(X_i)\le \overline p$ is
satisfied with equality.
In other words, the fraction $p(x)$ of observations with $X_i=x$ selected for
sampling the outcome $\tilde Y_i$ is proportional to the conditional standard
deviation $\sigma(x)=\sqrt{var(\tilde Y_i|X_i=x)}$ of the outcome, up to the
constraint that $p(x)\le 1$.
The prescription to sample proportionally to the conditional standard deviation
mirrors the formula for optimal allocation in
\citet[][p. 580]{neyman_two_1934} for sampling from a finite population.

\subsubsection{Other parameters of interest}\label{other_parameters_sec}

The general setting in Section \ref{multiple_treatments_general_results_sec}
allows for multiple treatments as well as outcomes that are transformations of
the original treatment variables.  As a simple example of the latter, one may
transform the outcome variable $Y_i$ in a way that one deems relevant to policy.
For example, one may wish to measure the success of a job training program in
terms of the proportion of individuals who obtain a wage above a certain level,
rather than simply using the average wage.
As discussed in Remark \ref{general_parameter_remark}, the results in Section
\ref{multiple_treatments_general_results_sec} apply
to other parameters of interest defined in terms of the
distributions of potential outcomes, so long as the efficient influence function
for the given parameter takes the same form as in Section \ref{multiple_treatments_general_results_sec}.
One example of such a parameter is the quantile treatment effect considered by
\citet{firpo_efficient_2007}.
In this setting, the treatment is binary and the object of interest is the
difference in the $\kappa$th quantile of the potential outcomes for some given $\kappa$:
\begin{align*}
  \Delta_\kappa=q_{\kappa,1}-q_{\kappa,0}
\end{align*}
where $q_{\kappa,w}$ denotes the $\kappa$th quantile of the distribution of
$Y_i(w)$ for $w=0,1$.
Under independent treatment assignment with propensity score $e(\cdot)$,
\citet{firpo_efficient_2007} shows that the efficient influence function is
given by
\begin{align*}
  \psi_\kappa(W_i,Y_i,X_i;e(\cdot))&=\frac{W_i}{e(x)}\left[ g_{1,\kappa}(Y_i)-E[g_{1,\kappa}(Y_i(1))|X_i] \right]
  -\frac{1-W_i}{e(x)}\left[ g_{0,\kappa}(Y_i)-E[g_{0,\kappa}(Y_i(0))|X_i] \right]  \\
  &+E[g_{1,\kappa}(Y_i(1))|X_i]
  -E[g_{0,\kappa}(Y_i(0))|X_i]
\end{align*}
where $g_{w,\kappa}(y)=-(I(y\le q_{w,\kappa})-\kappa)/f_{Y(w)}(q_{w,\kappa})$ for
$w=0,1$ where $f_{Y(w)}(y)$ denotes the marginal probability density function
of the potential outcome $Y_i(w)$ (p. 36 of the
supplemental appendix in \citet{firpo_efficient_2007}).
This efficient influence function and the least favorable submodel
used to obtain the semiparametric efficiency bound are the same as the ones for
estimating the parameter defined in Section
\ref{multiple_treatments_general_results_sec} with $a(x,y,1)=g_{1,\kappa}(y)$
and $a(x,y,0)=g_{0,\kappa}(y)$.
In particular, the score function for this submodel is given by
$\psi_\kappa(W_i,Y_i,X_i;e(\cdot))$ and the optimal asymptotic variance is
$v_{e(\cdot)}=E\psi_\kappa(W_i,Y_i,X_i;e(\cdot))^2$.
Solving for the choice of $e(\cdot)$ that minimizes the asymptotic variance
yields an analogue of the Neyman allocation:
\begin{align*}
  \frac{var(g_{1,\kappa}(Y_i(0))|X_i=x)}{[1-e^*(x)]^2}
  =\frac{var(g_{1,\kappa}(Y_i(1))|X_i=x)}{e^*(x)}.
\end{align*}

To verify the conditions of Corollary
\ref{general_asymptotic_minimax_corollary}, it therefore suffices to verify the differentiability
condition (\ref{tau_deriv_eq}) with $\tau(\cdot)$ 
given by the
quantile treatment effect as a function of the parameter $\theta$ in the
submodel.  This is also done on p. 36 of the supplementary appendix of \citet{firpo_efficient_2007} (see equation (A-13)
and the assertion afterward about the efficient influence function):
letting $\Delta_{\kappa}(\theta)$ denote the quantile treatment effect under
$\theta$ in a submodel around $\theta^*$ with score function given above, we have
\begin{align*}
  \Delta_{\kappa}(\theta^*+t)
  =tE_{\theta^*}\psi_\kappa(W_i,Y_i,X_i;e(\cdot))^2+o(t)
  =t v_{e(\cdot)}+o(t)
\end{align*}
where $v_{e(\cdot)}=E_{\theta^*}\psi_\kappa(W_i,Y_i,X_i;e(\cdot))^2$ is the
optimal asymptotic variance in this submodel.
Thus, under the conditions used to formalize the efficiency bound in
\citet{firpo_efficient_2007}, the conclusion of Corollary \ref{general_asymptotic_minimax_corollary}
holds with $v_{p^*()}$ given by the optimized variance in the setup above with
$a(x,y,w)$ given by $g_{w,\kappa}(y)$.

\section{Achieving the bound}\label{achieving_bound_sec}

While the main contribution of this paper is deriving lower bounds, we provide a
brief review here of approaches to achieving these bounds.  For ease of
exposition, we focus on the setting in Section \ref{binary_ate_formal_sec} of
estimating the average treatment effect in the setting of a binary treatment
without cost constraints.

As discussed in Section \ref{informal_description_sec}, a general strategy for
achieving the efficient asymptotic variance $v_{e^*(\cdot)}$ is to designate the
first $n_{\operatorname{pilot}}$ observations as a pilot study and use these
observations to form an estimate $\hat e^*_{\operatorname{pilot}}(\cdot)$ of the
optimal allocation defined in (\ref{neyman_allocation_eq}).  One then implements
this allocation in the remaining portion of the sample.  
When choosing the size of the pilot, the only formal requirements are that
$n_{\operatorname{pilot}}\to\infty$ (required for the estimate $\hat
e^*_{\operatorname{pilot}}(\cdot)$ to be consistent) and
$n_{\operatorname{pilot}}/n\to 0$ (required for the efficiency loss from
designating part of the sample for the pilot study to be asymptotically
negligible).
While any choice of $n_{\operatorname{pilot}}$ satisfying these conditions will
lead to the same asymptotic variance, considerations of finite sample
performance may lead to a more nuanced recommendation, a topic we discuss
briefly in Section \ref{recommendations_for_practice_sec}.

Regarding the choice of sampling scheme and estimator for the remainder of the
sample, there are two basic strategies.
The first is to assign treatment independently across observations, with observation $i$
receiving treatment $W_i=1$ with probability $\hat
e^*_{\operatorname{pilot}}(X_i)$ conditional on $X_i$.  One can then obtain an
efficient estimator using the remaining part of the sample\footnote{In practice, this estimator would be averaged with an
  estimator formed using the pilot observations
  $i=1,\ldots,n_{\operatorname{pilot}}$, but we leave this step out in this
  discussion for ease of exposition.} by using an estimator that achieves the
\citet{hahn_role_1998} bound.
For example, one can use the doubly robust estimator
\begin{align*}
  \hat \tau = \frac{1}{n-n_{\operatorname{pilot}}}\sum_{i=n_{\operatorname{pilot}+1}}^n
  \left\{  \frac{[Y_i-\hat \mu(X_i,1)]W_i}{\hat e^*_{\operatorname{pilot}}(x)}-\frac{[Y_i-\hat \mu(X_i,0)](1-W_i)}{1-\hat e^*_{\operatorname{pilot}}(x)}
  +\left[ \hat \mu(X_i,1)-\hat\mu(X_i,0) \right] \right\}
\end{align*}
where $\hat\mu(x,w)$ is an estimate of the conditional expectation
$E[Y_i(w)|X_i=x]=E[Y_i|X_i=x,W_i=w]$ of the potential outcome $Y_i(w)$.\footnote{The step of estimating the conditional expectation and incorporating it into the
estimator is needed only for efficiency: simply setting $\hat \mu(x,w)=0$ in the
definition above would
lead to a propensity score weighted estimator with a correctly specified
propensity score, which would yield a consistent estimate, but with a higher
asymptotic variance.}
Following the literature, we refer to this approach as ``regression
adjustment,'' since it involves ``adjusting'' a propensity score weighted
estimator using the regression estimate $\hat \mu(x,w)$.

The second strategy is to use stratification on the covariate $X_i$ when
assigning treatment to the remaining $n-n_{\operatorname{pilot}}$ units.
In the case where $X_i$ takes on a relatively small number of values, this can
be done
as follows.
Let
$N(x)$ be the
number of units $i=n_{\operatorname{pilot}}+1,\ldots,n$ with $X_i=x$ and let
$N_{\operatorname{treat}}^*(x)$ be obtained by rounding $\hat
e^*_{\operatorname{pilot}}(x)\cdot N(x)$ to the nearest integer.
One then chooses a subset of exactly $N_{\operatorname{treat}}^*(x)$ of the
units with $X_i=x$ to be treated, with the subset being chosen at random with
each of the ${N(x) \choose N_{\operatorname{treat}}^*(x)}$ possible subsets
having equal probability.
This leads to a treatment assignment with propensity score $\tilde
e^*(x)\equiv N_{\operatorname{treat}}^*(x)/N(x)$, but with dependence across
observations in the treatment decision, since setting $W_i=1$ for an observation
$i$ with $X_i=x$ means that other observations $j$ with $X_j=x$ are less likely
to receive treatment.
One can then obtain an efficient estimate using propensity score weighting:
\begin{align}\label{stratification_pscore_weighted_estimator_eq}
  \hat\tau = \frac{1}{n-n_{\operatorname{pilot}}}\sum_{i=n_{\operatorname{pilot}}+1}^n \left[ \frac{Y_iW_i}{\tilde e^*(x)}-\frac{Y_i(1-W_i)}{1-\tilde e^*(x)} \right].
\end{align}
In the case where $X_i$ is continuously distributed, one can implement
stratification by randomizing within bins (called ``strata'') of observations where
$X_i$ is nearly (but not exactly) the same.
One can then achieve efficiency by increasing the number of strata with the
sample size $n$ in such a way that the
variation in $X_i$ within each stratum decreases quickly enough as $n\to
\infty$, so long as certain smoothness conditions hold \citep[see][]{cytrynbaum_optimal_2023}.

Both of these approaches require some additional regularity conditions in order
to achieve the efficient asymptotic variance.  In the case of independent
treatment assignment, the estimator $\hat \mu(x,w)$ used to form the regression
adjustment should be consistent in order to achieve asymptotic efficiency
(although inconsistent estimates $\hat\mu(x,w)$ can still yield valid inference
due to the doubly robust nature of the estimator).
This means that one will need some smoothness conditions on $\mu(x,w)$ in order
to guarantee the consistency of a flexible nonparametric estimate.
Regularity conditions on $\mu(x,w)$ of a similar form are also needed to achieve
the efficient asymptotic variance when stratifying on continuous variables, in
order to rule out the possibility that the conditional means $\mu(X_i,0)$ and
$\mu(X_i,1)$ vary too much over $X_i$ in the same stratum.

Thus, while results in the literature show that the efficient
asymptotic variance can be achieved in large samples under sufficient
regularity using a number of different approaches, these results rely on
regularity conditions that may be
strong in settings where covariates $X_i$ are high dimensional or take on many
different values.
Furthermore,
asymptotic results may not provide an adequate description of
finite sample performance, particularly in the present setting where the
effective sample size is further decreased through the use of pilot studies.  In
the next section, we
turn these and other practical limitations and considerations of our results and
other results in the literature.

\section{Practical implications and limitations}\label{recommendations_for_practice_sec}

The asymptotic efficiency bounds derived in this paper provide a benchmark for
what can be achieved in sufficiently large samples with sufficient regularity
conditions.  However, there will often be practical limitations that will
make it undesirable to implement experimental designs where the optimal
allocation is estimated from pilot studies or initial waves of the experiment.

The overall size of the experiment and the complexity of the variables $X_i$
that one observes will often severely limit the scope for such designs.  One
must ask: is it really
feasible to use a small fraction of the sample to obtain a reliable estimate of
the Neyman allocation?
If not, one may choose to rely on prior beliefs about the variance of potential
outcomes to come up with a reasonable guess for the optimal allocation.
For example, in settings where treatment effects are expected to be small, one
will expect the conditional variance of $Y_i(w)$ to be similar for both
treatment groups, so that simply setting $e(x)=1/2$ (in the case of estimating
the average treatment effect with a single treatment and no constraints due to
cost of the treatment) will yield a design that is close to efficient.
\citet{cai_performance_2024} provide some discussion and Monte Carlo evidence on
this issue.

Another consideration is that the data from a single experiment may be used for
different purposes.  Since the efficient treatment allocation depends on the
parameter being estimated, one faces tradeoffs between tailoring the design of
the experiment to different goals.
For example, suppose that one is interested in comparing a baseline treatment
$W_i=0$ to two possible treatments, labeled $W_i=1$ and $W_i=2$.
If both of these treatments will be of interest for future policy choices, one
will want to balance the goals of accurate estimation of the treatment effects
$E[Y_i(1)-Y_i(0)]$ and $E[Y_i(2)-Y_i(0)]$.
The question of how treatment assignment should be chosen to balance the goals
of estimating different objects is outside of the scope of the optimality
framework considered here.

As discussed in Section \ref{achieving_bound_sec},
the asymptotic variance bound
in \citet{hahn_role_1998} can be achieved either 
through a combination of iid treatment assignment and regression adjustment,
or through more complicated treatment assignment schemes involving stratification.
The results in this paper confirm that no further asymptotic efficiency gain is
possible using stratification or even with more complex treatment assignment
schemes.
Given the asymptotic equivalence of these different approaches, one may
focus on other considerations when making choices about experimental
design.
For example, one may opt for stratification because the resulting efficient
estimator (the propensity score weighted estimator
(\ref{stratification_pscore_weighted_estimator_eq}))
can be easily described in a preanalysis plan,
making it easy to ``tie one's
hands'' regarding the estimator one will ultimately use.
Conversely, one may opt for iid sampling if practical considerations make
stratification difficult, or if one has prior beliefs about the structure of the
conditional mean $\mu(x,w)$ of potential outcomes that are more easily captured
in a regression specification rather than a stratification scheme.

\bibliography{../../../../library}

\appendix

\section{Proofs}

\subsection{Proof of Theorem \ref{lr_expansion_thm}}

It is immediate from %
Theorem 7.2 in \citet{van_der_vaart_asymptotic_1998} that
\begin{align*}
  \sum_{i=1}^n\tilde \ell_X(X_i;\theta_n)
  =\frac{1}{\sqrt{n}}\sum_{i=1}^nh's_X(X_i)
  - \frac{1}{2}h'I_Xh + o_{P_{\theta^*}}(1).
\end{align*}
To prove (\ref{lr_expansion_eq}), we obtain a similar decomposition for the terms involving $\tilde
\ell_{Y(w)|X}$.  Let $w\in\mathcal{W}$ be given.
Let $V_{n,i}=2\left[\frac{\sqrt{f_{Y(w)|X}(Y_i(w)|X_i;\theta_n)}}{\sqrt{f_{Y(w)|X}(Y_i(w)|X_i;\theta^*)}}-1\right]$.
The qmd condition then implies
$nE_{\theta^*}[(V_{n,i}-n^{-1/2}h's_w(Y_i(w)|X_i))^2]\to 0$.
Note that
\begin{align*}
\tilde \ell_{Y(w)|X}(Y_i,X_i;\theta)=2\log \left(1+\frac{1}{2}V_{n,i}\right)
=V_{n,i}  - \frac{1}{4}V_{n,i}^2 + V_{n,i}^2r(V_{n,i})
\end{align*}
where the last equality uses a second order Taylor expansion of $t\mapsto 2\log
(1+t/2)$, with $\lim_{t\to 0}r(t)=0$.
It follows immediately from the proof of
Theorem 7.2 in \citet{van_der_vaart_asymptotic_1998} that
$\sum_{i=1}^nI(W_{n,i}=w)V_{n,i}^2|r(V_{n,i})|\le
\sum_{i=1}^nV_{n,i}^2|r(V_{n,i})|=o_{P_{\theta^*}}(1)$.  Thus,
\begin{align*}
  \sum_{i=1}^n I(W_{n,i}=w)\tilde \ell_{Y(w)|X}(Y_i,X_i;\theta)
  =\sum_{i=1}^n I(W_{n,i}=w)V_{n,i}  - \frac{1}{4}\sum_{i=1}^n I(W_{n,i}=w)V_{n,i}^2
  + o_{P_{\theta^*}}(1).
\end{align*}
We will show that each of the terms
\begin{align}
  &\sum_{i=1}^nI(W_{n,i}=w)\left[V_{n,i}-E_{\theta^*}[V_{n,i}|X_i]-n^{-1/2}h's_w(Y_i(w)|X_i)\right]  \label{v_minus_expectation}  \\
  &\sum_{i=1}^nI(W_{n,i}=w)\left\{ E_{\theta^*}[V_{n,i}|X_i]+ \frac{1}{4n} h'I_{Y(w)|X}(X_i)h] \right\}  \label{v_expectation}  \\
  &\sum_{i=1}^nI(W_{n,i}=w)\left[  V_{n,i}^2 - \frac{1}{n}h'I_{Y(w)|X}(X_i)h \right]
    \label{v_squared}
\end{align}
converge in probability to zero under $\theta^*$.

Let
$A_{n,i}=V_{n,i}-E[V_{n,i}|X_i]-n^{-1/2}h's_w(Y_i(w)|X_i)$ 
so that the summand in (\ref{v_minus_expectation}) is given by
$I(W_{n,i}=w)A_{n,i}$.
For $i\le n$, let $\mathcal{F}_{2,n,i}$ denote the sigma algebra generated by
$X^{(n)}$, $\{Y_j(w)\}_{w\in\mathcal{W},1\le j\le i-1}$ and $U$.  Note that
$W_{n,i}$ is measurable with respect to $\mathcal{F}_{2,n,j}$ for $j\ge i$, and that
$A_{n,i}$ is measurable with respect to $\mathcal{F}_{2,n,j}$ for $j>i$.  In
addition,
$E_{\theta^*}[A_{n,i}|\mathcal{F}_{2,n,i}]=E_{\theta^*}[A_{n,i}|X_i]=0$, where
the last step uses the fact that $s_w$ is a score function conditional on $X_i$.
Thus, for $j>i$,
\begin{align*}
&E_{\theta^*}\left[ I(W_{n,i}=w)I(W_{n,j}=w)A_{n,i}A_{n,j}|\mathcal{F}_{2,n,j}
\right]
    = I(W_{n,i}=W_{n,j}=w)A_{n,i} E_{\theta^*}\left[ A_{n,j}|X_{j}
\right]=0
\end{align*}
so that the expectation of the square of (\ref{v_minus_expectation}) is given by
\begin{align*}
  \sum_{i=1}^n E_{\theta^*}I(W_{n,i}=w)A_{n,i}^2
  \le n E_{\theta^*}A_{n,i}^2
  \le n E_{\theta^*}\left\{ \left[ V_{n,i}-n^{-1/2}h's_w(Y_i(w)|X_i) \right]^2 \right\}
  \to 0
\end{align*}
by qmd, where the last inequality uses the fact that $A_{n,i}$ is equal to
$V_{n,i}-n^{-1/2}h's_w(Y_i(w)|X_i)$ minus its expectation given $X_i$.

For (\ref{v_expectation}), note that
\begin{align*}
  & E_{\theta^*}\left[ V_{n,i} |X_i\right]
  = E_{\theta^*}\left[ 2\frac{\sqrt{f_{Y(w)|X}(Y_i|X_i,\theta_n)}}{\sqrt{f_{Y(w)|X}(Y_i|X_i,\theta^*)}}-2 \bigg|X_i\right]  \\
  &= E_{\theta^*}\left[ 2\frac{\sqrt{f_{Y(w)|X}(Y_i|X_i,\theta_n)}}{\sqrt{f_{Y(w)|X}(Y_i|X_i,\theta^*)}}-\frac{f_{Y(w)|X}(Y_i|X_i,\theta_n)}{f_{Y(w)|X}(Y_i|X_i,\theta^*)} -1 \bigg|X_i\right]  \\
  &= -E_{\theta^*}\left[\left(  \frac{\sqrt{f_{Y(w)|X}(Y_i|X_i,\theta_n)}}{\sqrt{f_{Y(w)|X}(Y_i|X_i,\theta^*)}} - 1 \right)^2 \bigg|X_i\right]
    = -\frac{1}{4}E_{\theta^*}\left[ V_{i,n}^2 |X_i\right].
\end{align*}
Thus, the expectation of the absolute value of (\ref{v_expectation}) is bounded
by $1/4$ times
\begin{align*}
  &n E_{\theta^*}\left\{ \left| E_{\theta^*}\left[  V_{i,n}^2 - h'I_{Y(w)|X}(X_i)h/n  |X_i\right] \right| \right\}
    =  E_{\theta^*}\left\{ \left| E_{\theta^*}\left[ nV_{i,n}^2 - (h's_w(Y_i(w)|X_i))^2  |X_i\right] \right| \right\}.
\end{align*}
Letting $\tilde V_i=h's_w(Y_i(w)|X_i)$, this is bounded by
\begin{align*}
  &E_{\theta^*}\{ |nV_{i,n}^2 - [h's_w(Y_i(w)|X_i)]^2| \}
  = E_{\theta^*}(|n V_{i,n}^2 - \tilde V_i^2|)
  = E_{\theta^*}[|(\sqrt{n} V_{i,n} + \tilde V_i)(\sqrt{n} V_{i,n} - \tilde V_i)|]  \\
  &\le \sqrt{E_{\theta^*}[(\sqrt{n} V_{i,n} + \tilde V_i)^2]}\sqrt{E_{\theta^*}[(\sqrt{n} V_{i,n} - \tilde V_i)^2]}  \\
  &\le \left\{2\sqrt{E_{\theta^*}(\tilde V_i^2)}+\sqrt{E_{\theta^*}[(\sqrt{n} V_{i,n} - \tilde V_i)^2]}\right\}\sqrt{E_{\theta^*}[(\sqrt{n} V_{i,n} - \tilde V_i)^2]}.
\end{align*}
This converges to zero since
$E_{\theta^*}[(\sqrt{n} V_{i,n} - \tilde V_i)^2]=nE_{\theta^*}[(V_{i,n} -
n^{-1/2}h's_w(Y_i(w)|X_i))^2]\to 0$ by qmd.

For (\ref{v_squared}), note that
$E_{\theta^*}\left\{ \left| \sum_{i=1}^nI(W_{n,i}=w)\left[ V_{i,n}^2 - (n^{-1/2}h's_w(Y_i(w)|X_i))^2 \right] \right| \right\}$ 
is bounded by
$E_{\theta^*}\{ |nV_{i,n}^2 - [h's_w(Y_i(w)|X_i)]^2| \}$,
which was shown above to converge to zero.
Thus, to show that
(\ref{v_squared}) converges in probability to zero under $\theta^*$, it suffices
to show that
$\frac{1}{n}\sum_{i=1}^nI(W_{n,i}=w)\left[ (h's_w(Y_i(w)|X_i))^2 - h'I_{Y(w)|X}(X_i)h \right]$
converges in probability to zero under $\theta^*$.
This follows by a law of large numbers for martingale
difference arrays \citep[Theorem 2 in][]{andrews_laws_1988}, since
the summand is a martingale difference array with respect to the
filtration $\mathcal{F}_{2,n,i}$, and it is uniformly integrable under
$\theta^*$ since
it is bounded by the sequence $(h's_w(Y_i|X_i))^2+h'I_{Y(w)|X}(X_i)h$, which is
iid and has finite mean. 
This completes the proof of (\ref{lr_expansion_eq}).

\subsection{Proof of Corollary \ref{lr_clt_corollary}}

We use a martingale representation similar to the one used for matching estimators by
\citet{abadie_martingale_2012}.
For $i=1,\ldots,n$, let $\widetilde{\mathcal{F}}_{n,i}$ denote the sigma algebra
generated by $X_1,\ldots,X_i$, and let $B_{n,i}=h's_X(X_i)/\sqrt{n}$.  For $i=n+1,\ldots,
2n$, let $\widetilde{\mathcal{F}}_{n,i}$,
denote the sigma algebra generated by
$X^{(n)}$, $\{Y_j(w)\}_{w\in\mathcal{W},1\le j\le i-1-n}$ and $U$,
and let $B_{n,i}=\sum_{w\in\mathcal{W}}I(W_{n,i-n}=w)h's_{w}(Y_{i-n}(w)|X_{i-n})/\sqrt{n}$.
Then $\{B_{n,i}\}_{i=1}^{2n}$ is a martingale difference array with respect to
the filtration $\left\{\widetilde{\mathcal{F}}_{n,i}\right\}_{i=1}^{2n}$.  In addition,
$\sum_{i=1}^{2n}E_{\theta^*}[B_{n,i}^2|\widetilde{\mathcal{F}}_{n,i-1}]=h'\tilde I_n h$,
and, by Theorem \ref{lr_expansion_thm}, we have
$\ell_{n,h}=\sum_{i=1}^{2n}B_{n,i}-h'\tilde I_n h/2+o_{P_{\theta^*}}(1)$.
In case (i) (where $\tilde I_n$ converges in probability to
$\tilde I^*$ under $\theta^*$),
it then immediately from a central limit theorem for martingale
arrays \citep[Theorem 35.12][]{billingsley_probability_1995} that
$\ell_{n,h}$
converges
to a $N(-h'\tilde I^*h/2,h'\tilde I^*h)$ law under $\theta^*$ (the Lindeberg condition follows
since $\{B_{n,i}\}_{i=1}^n$ and $\{B_{n,i}\}_{i=n+1}^{2n}$ are each dominated by
sequences of iid variables with finite second moment).

Now consider case (ii), where $\tilde I_n\le \tilde I^*+o_{P_{\theta^*}}(1)$.  Let
$\Sigma_n=\Sigma_n(X^{(n)})$ be a sequence of positive semidefinite symmetric
matrices with $\tilde I_n+\Sigma_n=I^*+o_{P_{\theta^*}}(1)$.  Given
$U,X^{(n)},Y_{n}^{(n)}$, let $Z_1,\ldots,Z_n$ be iid and normally distributed
under $\theta$ with identity covariance and mean
$\Sigma_n^{1/2}(\theta-\theta^*)$.
Then
\begin{align*}
  &\tilde \ell_{n,h}=\log \frac{dP_{\theta^*+h/\sqrt{n}}}{dP_{\theta^*}}(U,X^{(n)},Y_n^{(n)},Z^{(n)})
  = \ell_{n,h} + \sum_{i=1}^n Z_i'\Sigma_n^{1/2}h/\sqrt{n}
  - h'\Sigma_nh/2  \\
  &= \sum_{i=1}^{2n}B_{n,i} + \sum_{i=1}^n Z_i'\Sigma_n^{1/2}h/\sqrt{n}
    -h'(\tilde I_n+\Sigma_n) h/2 + o_{P_{\theta^*}}(1)
\end{align*}
where the last step applies Theorem \ref{lr_expansion_thm}.
Let us define $B_{n,i}=Z_{i-2n}'\Sigma_n^{1/2}h/\sqrt{n}$ for $i=2n+1,\ldots
3n$, so that the above display can be written as $\sum_{i=1}^{3n}B_{n,i} - h'(\tilde I_n+\Sigma_n) h/2 +
o_{P_{\theta^*}}(1)$.  Letting
$\widetilde{\mathcal{F}}_{n,i}$ be the sigma algebra generated by
$\widetilde{\mathcal{F}}_{n,2n}$ and $Z_1,\ldots,Z_{i-2n}$ for
$i=2n+1,\ldots,n$, 
$\{B_{n,i}\}_{i=1}^{3n}$ is a martingale difference array with respect to
the filtration $\left\{\widetilde{\mathcal{F}}_{n,i}\right\}_{i=1}^{3n}$.
Furthermore, $\sum_{i=1}^{3n} E_{\theta^*}[B_{n,i}^2|\widetilde{\mathcal{F}}_{n,i-1}]= \tilde
h'(I_n+\Sigma_n)h=h'\tilde I^*h+o_{P_{\theta^*}}(1)$, and it satisfies the Lindeberg condition by the arguments above and uniform boundedness of $\Sigma_n$.
It therefore follows that $\tilde \ell_{n,h}$ converges in distribution under
$\theta^*$ to a $N(-h'\tilde I^*h/2,h'\tilde I^*h)$ law as claimed.

\subsection{Proof of Theorem \ref{ATE_LAN_thm}}

We have
$I_{Y(0)|X}(X_i)
= E[s_{Y(0)|X}(Y_i|X_i)^2|X_i]
= \frac{\sigma^2_{\theta^*}(X_i,0)}{[1-e^*(X_i)]^2}$
and 
$I_{Y(1)|X}(X_i)
= E[s_{Y(1)|X}(Y_i|X_i)^2|X_i]
= \frac{\sigma^2_{\theta^*}(X_i,1)}{e^*(X_i)^2}$
so that, by (\ref{neyman_allocation_eq}), $I_{Y(0)|X}(X_i)=I_{Y(1)|X}(X_i)$.
Letting $I_{Y|X}(X_i)=I_{Y(0)|X}(X_i)=I_{Y(1)|X}(X_i)$, we then have
\begin{align*}
  I_X + \frac{1}{n}\sum_{i=1}^n\sum_{w\in\{0,1\}}I(W_{n,i}=w) I_{Y(w)|X}(X_i)
  =  I_X + \frac{1}{n}\sum_{i=1}^n I_{Y|X}(X_i),
\end{align*}
which converges to $v_{e^*(\cdot)}$ under $\theta^*$ by the law of large numbers.
Thus, applying Corollary \ref{lr_clt_corollary}(i) with $v_{e^*(\cdot)}$ playing
the role of $\tilde I^*$, $\ell_{n,h}$ converges to a $N(-h^2
v_{e^*(\cdot)}/2,h^2 v_{e^*(\cdot)})$ law under $\theta^*$ as claimed.

\subsection{Proof of Corollary \ref{ate_asymptotic_minimax_corollary}}

The result is immediate from local asymptotic normality and the local asymptotic
minimax theorem, as stated in Theorem 3.11.5 in
\citet{van_der_vaart_weak_1996}.
(Formally, we consider the submodels
$\theta^*+\tilde h (n v_{e^*(\cdot)})^{-1/2}$ indexed by
$\tilde h$ when applying the definition of local asymptotic normality on p. 412.
Then 
$n^{1/2}[ATE(\theta^*+\tilde h (n v_{e^*(\cdot)})^{-1/2})-ATE(\theta^*)]\to
\tilde h v_{e^*(\cdot)}^{1/2}$, so that the derivative condition on the top of
p. 413 holds with $\dot\kappa(t)=v_{e^*(\cdot)}^{1/2}t$.)

\subsection{Proof of Theorem \ref{general_LAN_thm}}

We have
\begin{align*}
  &I_X+\frac{1}{n}\sum_{i=1}^n\sum_{w\in\mathcal{W}}I(W_{n,i}=w)I_{Y(w)|X}(X_i)
    =I_X+\frac{1}{n}\sum_{i=1}^n\sum_{w\in\mathcal{W}}I(W_{n,i}=w)[\lambda(X_i)+\mu' r(X_i,w)]  \\
  &\le I_X+\frac{1}{n}\sum_{i=1}^n(\lambda(X_i)+\mu' c) + o_{P_{\theta^*}}(1)
    = v_{p^*(\cdot)}+o_{P_{\theta^*}}(1)
\end{align*}
where the inequality uses (\ref{arbitrary_sampling_constraints_eq}) and the last
step applies the law of large numbers.
The result now follows from Corollary \ref{lr_clt_corollary}(ii), with
$v_{p^*(\cdot)}$ playing the role of $\tilde I^*$.

\subsection{Proof of Corollary \ref{general_asymptotic_minimax_corollary}}

The result is immediate from local asymptotic normality and the local
asymptotic minimax theorem \citep[][Theorem 3.11.5]{van_der_vaart_weak_1996}.
(As with the proof of Corollary \ref{ate_asymptotic_minimax_corollary}, we
consider the submodels $\theta^*+\tilde h (n v_{p^*()})^{-1/2}$ when applying
the definition of local asymptotic normality on p. 412.)

\end{document}